\documentclass[12pt,draftclsnofoot,onecolumn]{IEEEtran}
\usepackage{hyperref}
\usepackage{amsfonts}
\usepackage{algorithm}
\usepackage{array,color}
\usepackage{algpseudocode}
\usepackage{amssymb}
\usepackage{bm}
\usepackage{bbm}
\usepackage{booktabs}
\usepackage{cite,graphicx,amsmath,amsthm}
\usepackage{colortbl}
\usepackage{cases}
\usepackage{subfigure}
\usepackage{fancyhdr}
\usepackage{dsfont}
\usepackage{float}
\usepackage{makecell}
\usepackage{moreverb}
\usepackage{multirow}
\usepackage{tabularx}
\usepackage{wrapfig}
\usepackage{mathtools} 
\usepackage{gensymb}
\usepackage{mathtools}
\usepackage[numbers,sort&compress]{natbib}
\allowdisplaybreaks[4]

\newtheorem{lemma}{Lemma}

\newtheorem{proposition}{Proposition}

\newtheorem{corollary}{Corollary}

\newtheorem{property}{Property}

\newtheorem{remark}{Remark}

\newtheorem{claim}{Claim}

\begin{document}

\title{\huge Integrated Sensing and Communication for Edge Inference with End-to-End Multi-View Fusion}

\author{Xibin Jin, Guoliang Li, Shuai Wang, Miaowen Wen,\\Chengzhong Xu, \emph{Fellow, IEEE}, and H. Vincent Poor, \emph{Life Fellow, IEEE}  

\thanks{
Xibin~Jin and Miaowen~Wen are with the School of Electronic and Information Engineering, South China University of Technology, Guangzhou, China.
Guoliang~Li and Chengzhong~Xu are with the Department of Computer and Information Science, University of Macau, Macau, China.
Shuai~Wang is with the Shenzhen Institute of Advanced Technology, Chinese Academy of Sciences, Shenzhen, China.
H. Vincent Poor is with the Department of Electrical and Computer Engineering, Princeton University, Princeton, USA. 
Corresponding author: Shuai Wang ({\tt\footnotesize s.wang@siat.ac.cn}) and Miaowen Wen ({\tt\footnotesize eemwwen@scut.edu.cn}).
}
}
\maketitle

\begin{abstract}
Integrated sensing and communication (ISAC) is a promising solution to accelerate edge inference via the dual use of wireless signals. 
However, this paradigm needs to minimize the inference error and latency under ISAC co-functionality interference, for which the existing ISAC or edge resource allocation algorithms become inefficient, as they ignore the inter-dependency between low-level ISAC designs and high-level inference services.
This letter proposes an inference-oriented ISAC (IO-ISAC) scheme, which minimizes upper bounds on end-to-end inference error and latency using multi-objective optimization. The key to our approach is to derive a multi-view inference model that accounts for both the number of observations and the angles of observations, by integrating a half-voting fusion rule and an angle-aware sensing model. Simulation results show that the proposed IO-ISAC outperforms other benchmarks in terms of both accuracy and latency. 
\end{abstract}


\IEEEpeerreviewmaketitle

\section{Introduction}

\IEEEPARstart{R}{ealizing} perception intelligence at distributed wireless sensors is a challenge due to the conflict between computationally demanding deep neural networks (DNNs) and computationally limited sensors \cite{EI_Over}. Edge inference eases this conflict by providing DNN inference services to sensors with a proximal server \cite{Task-Ori}, \cite{MADDPG}. 
Existing edge inference consists of sensing, communication, and computation \cite{Task-Ori}.
In this letter, the sensing and communication stages are merged via the ISAC technique, leading to a two-stage edge inference solution, i.e., $\mathsf{ISAC}+\mathsf{Computation}$, thus a one-stage reduction compared to the existing solution.

In contrast to existing ISAC systems \cite{Edge_Acclerate}, ISAC edge inference systems aim to minimize the inference error instead of maximizing the sensing resolution or the communication data-rate. 
Therefore, the associated resource allocation becomes very different from traditional ISAC resource allocation that considers only the wireless channels.
For instance, the celebrated CRB-rate scheme allocates more resources to deterministic (random) channels for CRB (rate) maximization \cite{CRB_Rate}, and the multi-point ISAC scheme allocates more resources to closer-to-target devices to maintain the quality of sensing measurements \cite{MPISAC}. 
While these schemes have proven to be very efficient in traditional ISAC systems, they could lead to poor performance in edge inference systems (as illustrated in Fig.~\ref{System}), because they do not account for the inference requirements, e.g., the sensing data diversity that depends on the angle of observation. 
Note that existing edge inference schemes consider off-the-shelf sensing datasets, which ignore low-level sensing signal designs (e.g., physical-layer beamforming and power allocation) \cite{GNN_Edge} and their impact on high-level inference services.

This paper proposes an end-to-end approach that integrates sensing, communication, and computation objectives and constraints into a unified optimization framework.
First, to optimize the inference performance, we need an expression for the multi-view inference error with respect to the number of observations and the angles of observations. 
Hence, this paper derives an end-to-end inference error model by integrating a half-voting fusion rule and an angle-aware sensing model, and this newly derived model serves as an upper bound on the actual inference error.
Second, while ISAC inference eliminates one stage, it does not necessarily reduce the end-to-end inference latency, as ISAC also introduces interference between sensing and communication functionalities \cite{Aircomp}.
Therefore, this paper further derives the upper bound of the end-to-end sensing, communication, and computation latency with respect to the co-functionality interference. 
Third, to optimize the two conflicting objectives. i.e., inference error and inference latency, 
we minimize their upper bounds according to majorization minimization (MM), and propose a joint power allocation and device scheduling (JPADS) algorithm, leading to an inference-oriented ISAC (IO-ISAC) scheme.
Finally, various simulation results demonstrate the effectiveness of our proposed scheme, and show that the proposed IO-ISAC outperforms other benchmark schemes in terms of the inference accuracy and latency. 
To the best of our knowledge, this represents the first inference oriented paradigm guiding the design of ISAC systems. 

\section{System Model}
\begin{figure*}[!t]
\centering
	\vspace{-0.2cm}
	\subfigure[]{
		\label{System}
		\includegraphics[width=0.9\linewidth]{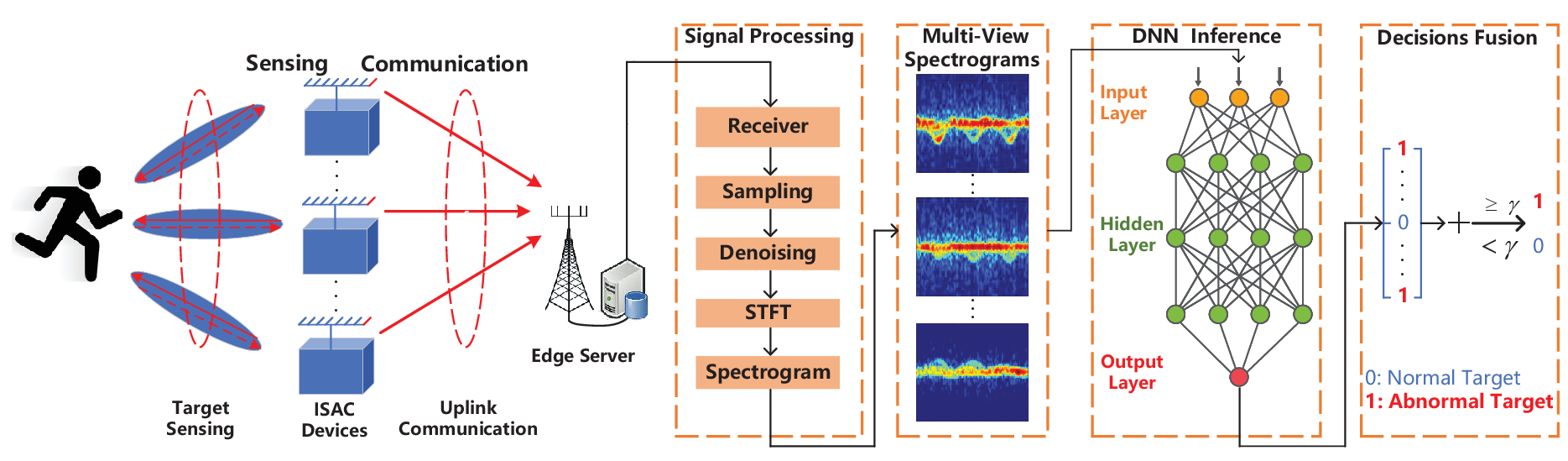}}
	\subfigure[]{
		\label{Arc}
		\includegraphics[width=0.45\linewidth]{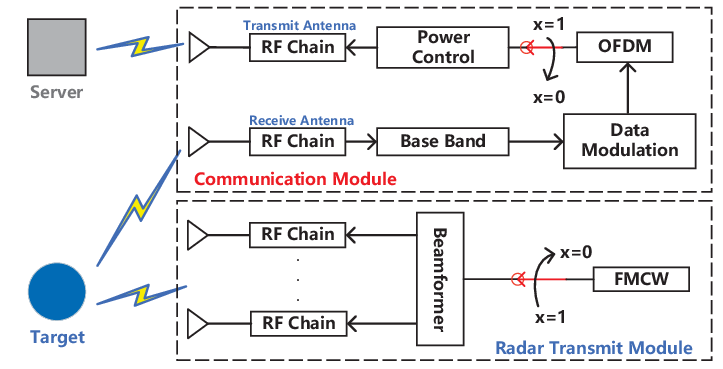}}
  	\subfigure[]{
  		\label{Sensing_Range}
		\includegraphics[width=0.45\linewidth]{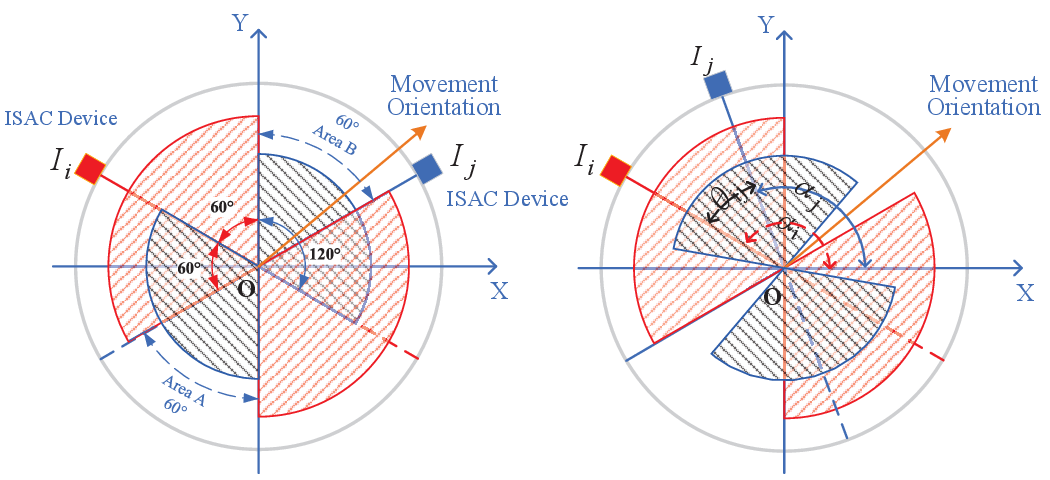}}
  \vspace{-0.1in}
	\caption{System model. (a) ISAC edge inference system with $N$ ISAC devices and 1 server; (b) Architecture of ISAC device; (c) Angle-aware sensing model. The target is at the center. The red and blue regions denote the effective sensing fields of devices $i$ and $j$, respectively.} 
	\label{Entire_1}
	\vspace{-0.2in}
\end{figure*}

We consider an ISAC edge inference system shown in Fig.~\ref{System}, which consists of $N$ ISAC devices each with $N_{d}$ antennas and one edge server with $N_{e}$ antennas. 
The task of our system is to detect the state of a target (i.e., normal or abnormal) using two consecutive stages.
In the first ISAC stage, ISAC devices transmit radio signals to sense the target and upload the echo sensing data to the server simultaneously at the same frequency band.
In the second computation stage, the server extracts features from the sensing data using DNNs and fuses the local inference results into a global decision.
Below we present the model for the ISAC stage.

In particular, Fig.~\ref{Arc} demonstrates the architecture of an ISAC device, which consists of $N_{t}$ transmit antennas and $N_{r}$ receive antennas, where $N_{t}+N_{r}=N_{d}$.
The transmit antennas are divided into two groups \cite{Antenna}, with $N_{s}$ antennas for sensing and the other $N_{c}$ antennas for communication. 
Without loss of generality, it is assumed that 
$N_r=N_c=1$.
To mitigate the interference leakage and energy cost, an activity controller with binary output $x_{i}\in\{0,1\}$ is adopted to determine the working state of the $i$-th ISAC device, where $x_{i}=1$ represents that the $i$-th ISAC device is activated and $x_{i}=0$ otherwise. 
Let $\mathbf{x}=[x_1,\ldots,x_N]^T\in\{0,1\}^N$ denote the states of all ISAC devices and $\mathcal{S}=\{i:x_{i}=1\}$ denote the set of activated ISAC devices with cardinality $|\mathcal{S}|$.

\subsubsection{Target Sensing Model}
The $i$-th ISAC device sends a sensing signal to probe the target and the echo signal received at the $i$-th ISAC device is
\begin{align}
	\!\!\!
	r_{i}&=
	x_{i}\mathbf{g}_{i}^{H}\mathbf{f}_{i}s_i+
	\sum\limits_{j \neq i}x_{j}\left[\mathbf{q}_{ji}^{H}\mathbf{f}_{j}s_{j}+\sqrt{p_{j}^{c}}c_{ji}d_j\right]
	+n_s, \label{r_s}
\end{align}
where $i, j \in \mathcal{S}$, $s_i\in \mathbb{C}$ with $\mathbb{E}\left[|s_i|^2\right]=1$ is the sensing signal, $d_{i}\in \mathbb{C}$ with $\mathbb{E}\left[|d_{i}|^2\right]=1$ is the communication signal, and $n_{s} \sim \mathcal{CN}(0,\sigma_{s}^2)$ is additive white Gaussian noise (AWGN) with power $\sigma_{s}^2$. Vector $\mathbf{g}_{i} \in \mathbb{C}^{N_{s}}$ denotes the sensing echo channel from the $i$-th device to the target and then back to the $i$-th device. The vectors $\mathbf{q}_{ji} \in \mathbb{C}^{N_{s}}$ and $\mathbf{f}_{j} \in \mathbb{C}^{N_{s}}$ denote the sensing interference channel from the $j$-th device to the $i$-th device and the sensing beamformer of the $j$-th device, respectively. The scalars $p_{j}^{c}$ and $c_{ji} \in \mathbb{C}$ represent the communication transmit power of the $j$-th device and the communication interference channel from the $j$-th device to the $i$-th device, respectively. 
We adopt zero-forcing beamforming $\mathbf{f}_i=\sqrt{p_{i}^{s}}\,\mathbf{f}_i^{\mathrm{ZF}}$ to mitigate sensing interference among different devices \cite{ISAC_overview2}, where $p_{i}^{s}$ and $\mathbf{f}_{i}^{\mathrm{ZF}}$ denote the transmit sensing power and the steering vector for interference cancelation of the $i$-th device. The vector $\mathbf{f}_{i}^{\mathrm{ZF}}$ is the normalized $1$-st column of $\mathbf{F}_i^{H}\left(\mathbf{F}_i\mathbf{F}_i^{H}\right)^{-1}$, 
where $\mathbf{F}_i=\left[\mathbf{g}_{i}~\mathbf{q}_{i1}\dots\mathbf{q}_{ik}\dots\right]^H$ for $k \in \mathcal{S}$ and $i \neq k$.
The sensing signal-to-interference-plus-noise ratio (SINR) of the $i$-th ISAC device is
\begin{align}
&\text{SINR}_i^{\mathrm{s}}=\frac{x_{i}p_{i}^{s}|\mathbf{g}_{i}^{H}\mathbf{f}_{i}^{\mathrm{ZF}}|^2}
{\sigma_{s}^2 + \sum_{j\neq i}x_jp_{j}^{c}|c_{ji}|^{2}}.
\end{align}

\emph{Remark 1}: The echo signal $r_i$ will be fed to the processor module for generating micro-doppler spectrograms shown in Fig.~\ref{System}. As such, the target sensing problem is transformed into an image classification problem that can be solved by DNNs \cite{HMR_Range}. 
Note that if the sensing SINR is below a certain threshold, the spectrogram quality and the DNN accuracy would be significantly degraded as shown in Fig.~5 of \cite{JSTSP}.
Such degradation also emerges if the sensing angle exceeds a certain threshold \cite{HMR_Range}, as shown in Fig.~\ref{Sensing_Range}.

\subsubsection{Uplink Communication Model}
The $i$-th ISAC device also transmits the communication signal for uploading the echo sensing data to the edge server, and the associated received signal at the server is \cite{Complex_Matrix}
\begin{align}
\mathbf{y}_i&=x_{i}\sqrt{p_{i}^{c}}\mathbf{h}_{i}d_{i}+\sum_{j}x_{j}\mathbf{O}_{j}\mathbf{f}_{j}s_{j}
+\sum_{j \neq i}x_{j}\sqrt{p_{j}^{c}}\mathbf{h}_{j}d_{j}+\mathbf{n}_{c}, 
\end{align}
\noindent where $i, j \in \mathcal{S}$, {the vector} $\mathbf{h}_{j} \in \mathbb{C}^{N_{e}}$ is the communication channel from the $j$-th device to the server. The matrix $\mathbf{O}_{j} \in \mathbb{C}^{N_{e} \times N_{s}}$ represents the sensing interference channel from the $j$-th device to the edge server. {The vector} $\mathbf{n}_{c} \sim \mathcal{CN}(\mathbf{0},\sigma_{c}^2\mathbf{I}_{N_e})$ is the AWGN. 
Since the transmitted sensing signal is deterministic, e.g., frequency modulated continuous wave (FMCW), and the channels $\{\mathbf{f}_{j},\mathbf{O}_{j}\}$ are known at the server, the sensing interference term can be eliminated \cite{eliminate}. The server exploits a receiving filter $\mathbf{w}_{i} \in \mathbb{C}^{N_{e}}$ to recover the transmit signal of the $i$-th device. The uplink signal of the $i$-th device recovered at the server is $y_{i}\!=\!\mathbf{w}_{i}^{H}\mathbf{y}\!$, i.e.,
$y_{i}\!= x_{i}\!\sqrt{p_{i}^{c}}\mathbf{w}_{i}^{H}\mathbf{h}_{i}d_{i}\!+\!\sum_{j \neq i}\!x_{j}\!\sqrt{p_{j}^{c}}\mathbf{w}_{i}^{H}\mathbf{h}_{j}d_{j} 
	     \!+\!\mathbf{w}_{i}^{H}\mathbf{n}_{c}$.
Here we adopt zero-forcing receive filters $[\mathbf{w}_{1}^{H},\mathbf{w}_{2}^{H},...,\mathbf{w}_{k}^{H}]^{T}=(\mathbf{H}^{H}\mathbf{H})^{-1}\mathbf{H}^{H}$ such that $\mathbf{w}_{i}^{H}\mathbf{h}_{j}=0$ for $\forall i \neq j$, where $\mathbf{H}=[\mathbf{h}_{1},\ldots,\mathbf{h}_{k},\ldots]$, $\forall k \in \mathcal{S}$.
The associated uplink communication SINR of the $i$-th device is
\begin{align}
	\text{SINR}_{i}^{\mathrm{c}}=\frac{x_{i}p_{i}^{c}|\mathbf{w}_{i}^{H}\mathbf{h}_{i}|^{2}}
	{\|\mathbf{w}_{i}^{H}\|_2^{2}\sigma_{c}^2}.
\end{align}


\section{Inference Oriented ISAC}

This section presents the end-to-end model and algorithms by integrating the $\mathsf{ISAC}$ and $\mathsf{Comput-}$ 
$\mathsf{ation}$ stages.
First, each DNN at the server recognizes the micro-Doppler spectrograms generated by active devices as shown in Fig.~\ref{System}. Then multiple DNN inference results are fused into a global inference output via voting. The inference error $1-\Phi$ (i.e., $\Phi$ denotes inference accuracy) and latency $\Xi$ are functions of activation switch $\mathbf{x}$, sensing signal powers $\{p_i^s\}$, and communication signal powers  $\{p_i^c\}$.
The inference oriented problem is formulated as a multi-objective optimization problem of minimizing the error and the latency of inference, under individual and total transmit power constraints:
\begin{subequations}
	\begin{align}
	\mathrm{(P0)} \mathop{\mathrm{min}}_{\substack{\mathbf{x},\{p_{i}^{s}\},\{p_{i}^{c}\}}}~
	&\left(1-\Phi\left(\mathbf{x},\{p_{i}^{s}\},\{p_{i}^{c}\}\right),\,\Xi\left(\mathbf{x},\{p_{i}^{s}\},\{p_{i}^{c}\}\right)\right) \nonumber \\
	\mathrm{s.t.}~~~~~
        &x_i \in \left\{0,1\right\}, ~\forall i, \label{P01}\\
 	&p_{i}^{s}+ p_{i}^{c} \leq P_{\max}, ~\forall i, \label{P02}\\
	&\sum_{i} (p_{i}^{s} +p_{i}^{c} ) \leq P_{\rm{sum}}, ~ i \in \mathcal{S}, \label{P03}
	\end{align}
\end{subequations}

\noindent where $P_{\rm{sum}}$ and $P_{\max}$ represent the maximum transmit power of the system and of each ISAC device, respectively.
The challenge to solve $\mathrm{P0}$ is that the functions $(\Phi,\Xi)$ have no explicit expressions. To address this issue, Sections III-A and III-B will derive bounds for $(\Phi,\Xi)$, and Section III-C optimizes the bounding functions according to MM.

\subsection{End-to-End Inference Error Model}
We first derive the inference accuracy when $|\mathcal{S}|\!=\!1$. 
Without loss of generality, consider the $i$-th ISAC device and its inference accuracy $\Phi_{i}$ is modeled as
\begin{align}
\Phi_i&=P(H_0)(1-P(B_{i}=1|H_0))
+P(H_1)(1-P(B_{i}=0|H_1)),
\end{align}
\noindent where $P(\cdot)$ represents the probability function, $H_0$ and $H_1$ are hypotheses that represent ``normal target'' and ``abnormal target'', $P(H_0)$ and $P(H_1)$ are the prior probabilities of the two hypotheses, and $B_{i}$ represents the binary inference result of the $i$-th device, with $B_{i}=0$ standing for ``normal target'' and $B_{i} = 1$ standing for ``abnormal target''. 
Now we consider two factors that affect $\Phi_i$ as mentioned in \emph{Remark 1}. First, $\Phi_i$ depends on the $\text{SINR}_{i}^{s}$. If $\text{SINR}_{i}^{s} \geq \beta$, where $\beta$ denotes the sensing quality threshold, then a satisfactory $\Phi_i$ can be achieved; otherwise, $\Phi_i$ would be degraded by a certain factor. According to \cite{JSTSP}, we set $\beta=27\,\text{dB}$.
Second, $\Phi_i$ depends on the observation angle $\xi_i$, which is defined as the aspect angle between the line-of-sight direction and the target motion direction.
If $|\mathrm{cos}(\xi_i)|\in[\alpha,1]$, where $\alpha$ denotes the angle threshold, a satisfactory $\Phi_i$ can be achieved; otherwise, $\Phi_i$ would be degraded by a certain factor. According to \cite{HMR_Range}, we set $\alpha=0.5$.
Based on the above, the false alarm probability is modeled as
\begin{align}
& P(B_{i}=1|H_0)=
\left \{
		\begin{array}{lr}
			P_{f},~~~\mathrm{if}~\text{SINR}_{i}^{s} \geq \beta,~|\mathrm{cos}(\xi_i)|\geq\alpha\\
			\lambda P_{f},~\mathrm{otherwise}  \\
		\end{array}
		\right., \nonumber
\end{align}
where $P_f$ is a constant value and $\lambda>1$ is the degradation factor.
Similarly, the missed detection probability is given by $P(B_{i}\!=\!0|H_1)\!=\!P_{m}$ if the SINR and observation angle requirements are satisfied; and $P(B_{i}\!=\!0|H_1)\!=\! \lambda P_{m}$ otherwise, where $P_m$ is also a constant.
The values of $(P_{f},P_{m},\lambda)$ can be estimated from the experimental data \cite{JSTSP, MPISAC}.

With the single-device model, we can derive the multi-view inference accuracy model $\Phi$.
Let $\gamma$ denote the voting threshold\footnote{The most widely adopted threshold is $\gamma=\lfloor \frac{|\mathcal{S}|}{2} \rfloor$, which corresponds to the half-voting rule for ensemble inference.} such that our global inference after fusion is
$\hat{H}_0: \sum_{i\in\mathcal{S}} B_{i} < \gamma$ and $\hat{H}_1: \sum_{i\in\mathcal{S}} B_{i} \geq \gamma$.
Following similar derivations in \cite{MPISAC}, the inference accuracy is given by 
$\Phi(\mathbf{x},\{p_i^s\},\{p_i^c\})
=\Theta(\mathbf{x},\{p_i^s\},\{p_i^c\}; Z)$, where $Z$ is the actual number of ISAC devices satisfying the SINR and angle requirements, and
\begin{align}\label{Theta}
	\begin{array}{*{35}{l}}
		{} & \Theta (\mathbf{x},\{p_{i}^{s}\},\{p_{i}^{c}\};z)=P({{H}_{0}})\sum\limits_{n=\gamma }^{|S|}{}\sum\limits_{k={{k}_{1}}}^{{{k}_{2}}}{}\left( \begin{smallmatrix}
			z \\
			k
		\end{smallmatrix} \right)\left( \begin{smallmatrix}
			|\mathcal{S}|-z \\
			n-k
		\end{smallmatrix} \right){{P}_{1}}  \\
		{} & \quad \quad \quad \quad \quad \quad +P({{H}_{1}})\sum\limits_{n=|S|-\gamma +1}^{|S|}{\sum\limits_{k={{k}_{1}}}^{{{k}_{2}}}{}}\left( \begin{smallmatrix}
			z \\
			k
		\end{smallmatrix} \right)\left( \begin{smallmatrix}
			|\mathcal{S}|-z \\
			n-k
		\end{smallmatrix} \right){{P}_{2}},  \\
	\end{array}
\end{align}

where $k_1=\max\{0,n-|\mathcal{S}|+z\}$, $k_2=\min\{z, n\}$, and\footnote{Due to the different locations and observation views, the sensing data at different ISAC devices is assumed to be independent.}
\begin{align}
	\begin{split}
		\left \{
		\begin{array}{lr}
			P_1=(1-P_f)^k (1-\lambda P_f)^{n-k} P_{m}^{z-k} (\lambda P_m)^{|\mathcal{S}|-z-n+k}  \\
            \\ \nonumber
            P_2=(1-P_m)^k (1-\lambda P_m)^{n-k} P_{f}^{|\mathcal{S}|-z-k} (\lambda P_f)^{|\mathcal{S}|-z-n+k}
		\end{array}
		\right. .
	\end{split}
\end{align}
\noindent To compute $\Phi$ using \eqref{Theta}, we still need to compute $Z$.
However, $Z$ is random since $\{\xi_i\}$ depends on the unknown target motion.
Fortunately, $\Theta(\mathbf{x},\{p_i^s\},\{p_i^c\}; z)$ is a monotonically increasing function of $z$ due to the property of binomial functions.
Therefore, we can bound $Z$ from below to get a lower bound of $\Phi$.
Specifically, let $\mathcal{Q}(\mathbf{x},\{p_i^s\},\{p_i^c\})=\{i\in\mathcal{S}:\text{SINR}_i^s\geq\beta\}$, and denote
the intersection angle between the $i$-th and $j$-th devices with $i,j\in\mathcal{Q}$ as
$\theta_{i,j}$ ($0\leq\theta_{i,j}\leq \pi$). 
Then we construct matrix $\mathbf{T}(\mathbf{x},\{p_i^s\},\{p_i^c\})$ , with the element $(i,j)$ being
\begin{align}
	\begin{split}
		T_{ij}= \left \{
		\begin{array}{lr}
			1, & \mathrm{if} \quad \frac{\pi}{3} \leq \theta_{ij} \leq \frac{2\pi}{3} \ \mathrm{and} \ i < j  \\
			-1,& \mathrm{if} \quad \frac{\pi}{3} \leq \theta_{ij} \leq \frac{2\pi}{3} \ \mathrm{and} \ i > j \\
			0, &\mathrm{otherwise}
		\end{array}
		\right.,
	\end{split}
\end{align}
and the following proposition is established.
\begin{proposition}
The global inference accuracy $\Phi$ satisfies 
\begin{align}
\Phi(\mathbf{x},\{p_i^s\},\{p_i^c\})
\!\geq\!\underbrace{\Theta\left
(\!\mathbf{x},\{p_i^s\}\{p_i^c\};\frac{1}{2}\mathrm{Rank}(\mathbf{T})\!\right)\!}_{:=\Phi_{\mathsf{LB}}}.
\end{align}
\end{proposition}

\begin{proof}
To prove the proposition, we only need to show that $\frac{1}{2}\mathrm{Rank}(\mathbf{T})\leq L$ holds.
We adopt the induction method to prove it, and first consider the two-devices case shown in Fig.~\ref{Sensing_Range}.
\begin{itemize}
\item If $\frac{\pi}{3} \leq \theta_{ij} \leq \frac{2\pi}{3}$, then we have $\mathbf{T}=[0,1;-1,0]$ and $\frac{1}{2}\mathrm{rank}(\mathbf{T})=1$.
In such a case, the device $j$ always falls within the union of areas A and B as shown in the left hand side of Fig.~\ref{Sensing_Range}. Hence we must have $\exists l\in\{i,j\},~|\mathrm{cos}(\xi_l)|\geq\frac{1}{2}$, and the target is guaranteed to be sensed by at least one ISAC device no matter in which direction the target moves. 
This implies that $L\geq 1$, and the $\frac{1}{2}\mathrm{Rank}(\mathbf{T})\leq L$ holds.
\item
If $\theta_{ij}$ is not in the range of $[\frac{\pi}{3},\frac{2\pi}{3}]$, then we have $\mathbf{T}=[0,0;0,0]$ and $\frac{1}{2}\mathrm{rank}(\mathbf{T})=0$.
In such a case, there always exist sensing blind regions as shown in the right hand side of Fig.~\ref{Sensing_Range}, and $\frac{1}{2}\mathrm{Rank}(\mathbf{T})\leq L$ also holds.
\end{itemize}

Next, we prove the $N$-devices case via graph theory.
We use $\mathcal{G}=<\mathcal{V},\mathcal{E}>$ to represent an undirected graph, where $\mathcal{V}$ is the set of nodes and $\mathcal{E}$ is the set of edges between any two nodes in $\mathcal{V}$. 
Let $\mathcal{V}=\mathcal{S}$, and connect nodes $i$ and $j$ to form edge $E=(i,j)$ if $i,j \in S$ and $\frac{\pi}{3} \leq \theta_{ij} \leq \frac{2\pi}{3}$;
otherwise, nodes $i$ and $j$ are not directly connected. 
For graph $\mathcal{G}$, 
we define $\mathcal{E}^\diamond$ as a matching of $\mathcal{G}$ if $\mathcal{E}^\diamond\in \mathcal{E}$ and any two distinct edges in $\mathcal{E}^\diamond$ have no common endpoints. 
The size of $\mathcal{E}^\diamond$, denoted as $S(\mathcal{E}^\diamond)$, is defined as the number of edges in $\mathcal{E}^\diamond$.
Then the maximum matching of $\mathcal{G}$, which is denoted by $\max S(\mathcal{E}^\diamond)$, represents the minimum number of ISAC devices with error probabilities $(P_f,P_m)$ no matter in which direction the target moves. 
This means that $L\geq \max S(\mathcal{E}^\diamond)$ holds. 
On the other hand, according to the \cite{max-matching}, we have $\max S(\mathcal{E}^\diamond)=\frac{1}{2}\mathrm{Rank}(\mathbf{T})$. 
Therefore $L\geq \frac{1}{2}\mathrm{Rank}(\mathbf{T})$, and this completes the proof. 
\vspace{-3pt}
\end{proof}


\vspace{-15pt}
\subsection{End-to-End Inference Latency Model}
Mathematically, the end-to-end inference latency is equal to 
$\Xi=T_{\mathsf{ISAC}}+T_{\mathsf{Comp}}$.
Since sensing and communication are executed simultaneously in the ISAC stage, we have $T_{\mathsf{ISAC}}=\max\{T_{s}, T_{c}\}$, where $T_{s}$ is the time for each ISAC device to sense the target and $T_{c}$ is the time for uploading sensing data to the server. 
Now we consider two cases: (i) $T_{s}\geq T_c$, which happens if the target motion is slow and the communication speed is fast; (ii) $T_{s}\leq T_c$, which happens if the target motion is fast and the communication speed is slow. 
In the considered system, due to co-channel interference brought by multiple working devices, the communicaion time is typically larger than the sensing time, so we have $T_{\mathsf{ISAC}}=T_{c}$. Since the server needs to wait for the last data to be received for subsequent fusion, $T_{c}$ is upper bounded by the device with the worst channel, which is given by
$T_{c}\leq\max\limits_{i \in S} \frac{D}{BR_{i}(\mathbf{x},p_{i}^c)}.$
Where $D$ and $B$ represent the data volume per sample and the bandwidth for each device, and $R_{i}$ is the uplink spectral efficiency of the $i$-th device in bps/Hz, which is written as
\begin{align}\label{spectral efficiency}
	R_{i}(\mathbf{x},p_{i}^{c})=x_{i}\log_{2}\left(1+\frac{p_{i}^{c}|\mathbf{w}_i^{H}\mathbf{h}_i|^2}
	{\|\mathbf{w}_{i}^{H}\|_2^{2}\sigma_{c}^2}\right).
\end{align}
\indent Moreover, the model computation time $T_{\mathsf{Comp}}$ includes signal processing and DNN inference time, which is $T_{\mathsf{comp}}=|S|N_{\mathsf{flop}}/f$, where $N_{\mathsf{flop}}$ and $f$ represent the number of floating-point operations needed for processing each data sample and computing speed (FLOPs/s) at the server, respectively. The total latency for one inference round is bounded by
\begin{align}
	\Xi\left(\mathbf{x},\{p_{i}^{s}\},\{p_{i}^{c}\}\right)&
 \leq
 \underbrace{
 \max\limits_{i \in S} \frac{D}{BR_{i}(\mathbf{x},\{p_{i}^{c}\})}+\frac{|S|N_{\mathsf{flop}}}{f}}_{:=
 \Xi_{\mathsf{UB}}\left(\mathbf{x},\{p_{i}^{s}\},\{p_{i}^{c}\}\right)}.
\end{align}

\vspace{-0.25in}
\subsection{MM Method and JPADS Algorithms}
Based on the results in Sections III-A and III-B, we adopt the MM technique to replace 
functions $(\Phi,\Xi)$ with $(\Phi_{\mathrm{LB}},\Xi_{\mathrm{UB}})$. 
Such a procedure would guarantee the worst-case performance of the ISAC edge inference system. Then problem $\mathrm{P0}$ is transformed into 
\begin{subequations}
	\begin{align}
	(\mathrm{P}1)\mathop{\mathrm{min}}_{\substack{\mathbf{x},\{p_{i}^{s}\},\{p_{i}^{c}\}}}~
	&\!\Big(1\!-\!\Phi_{\mathsf{LB}}(\mathbf{x},\{p_{i}^{s}\},\{p_{i}^{c}\}), 
 \Xi_{\mathsf{UB}}(\mathbf{x},\{p_{i}^{s}\},\{p_{i}^{c}\})\!\Big)\nonumber \\
	\mathrm{s.t.}~~~ &\text{SINR}_i^{\mathrm{s}} \geq \beta,~\forall i \in \mathcal{S}, \label{P101} \\
	&(6a)-(6c). \label{P102}
	\end{align}
\end{subequations}

This is a multi-objective mixed integer nonlinear programming (MO-MINLP) problem, and finding its Pareto optimal solutions is NP-hard in general.
However, $\mathrm{P}1$ enjoys a special property that enables us to achieve Pareto optimality if the number of ISAC devices is small. 
In particular, 
given arbitrary device activation $\{\mathbf{x}=\mathbf{x}'\}$, problem $\mathsf{P}1$ with respect to $\{p_{i}^{s},p_{i}^{c}\}$ is a convex problem
	\begin{align}
		(\mathsf{P}2):\mathop{\mathrm{max}}_{\substack{\{p_{i}^{s}\},\{p_{i}^{c}\}}}~~
		&\min_{\forall i \in \mathcal{S}'}~{\frac{p_{i}^{c}|\mathbf{w}_i^{H}\mathbf{h}_i|^2}{\Vert \mathbf{w}_i^{H} \Vert_{2}^2\sigma_{c}^{2}}} \nonumber\\
		\mathrm{s.t.}~~~ &(6b), (6c), \text{SINR}_i^{\mathrm{s}} \geq \beta, ~\forall i \in \mathcal{S}',
	\end{align}
which can be solved by CVX with a computational complexity of $\mathcal{O}(N^{3.5})$.
Note that to obtain $\mathsf{P}2$ from $\mathsf{P}1$, we have removed the terms not related to $\{p_{i}^{s},p_{i}^{c}\}$, and leveraged the monotonicity of logarithm function.
With the above property, the set of all Pareto optimal solutions can be found via a two-tier algorithm, with the outer layer explores all possible values of $\mathbf{x}$ using exhaustive tree search, and the inner layer solving $\mathsf{P}2$ with a fixed $\mathbf{x}$ controlled by the outer layer. 
This algorithm is named optimal-JPADS.

In practice, the number of ISAC devices may be large, and the optimal-JPADS algorithm would be time consuming. 
To this end, we further propose an iterative local search method, which is much faster and empirically achieves a solution set close to the Pareto optimality (as shown by experiments in Fig.~3b of Section IV).\footnote{The optimal solution to the weighted-sum problem of $\mathsf{P}1$ is not guaranteed to be a Pareto optimal solution \cite{WeakPareto}.}
First, we start from a feasible solution of $\mathbf{x}$, denoted as $\mathbf{x}^{[0]}$, and defines the neighborhood of $\mathbf{x}^{[0]}$ as:
\begin{align}
	\mathcal{N}(\mathbf{x}^{[0]})=\{\mathbf{x}:||\mathbf{x}-\mathbf{x}^{[0]}||_0\leq L,~x_i\in\{0,1\}, \forall i\},
\end{align}
which $1 \leq L \leq N$ is the variable size. 
Second, we randomly flip $L$ elements of $\mathbf{x}^{[0]}$ to generate a new feasible solution $\mathbf{x}' \in \mathcal{N}(\mathbf{x}^{[0]})$. Third, with the choice of $\mathbf{x}$ fixed to $\mathbf{x}=\mathbf{x}'$, we solves $\mathrm{P2}$ and obtains the associated optimal solution $\{p_{i}^{s'},p_{i}^{c'}\}$.
Fourth, we compute the weighted-sum function $\Psi(\mathbf{x}',\{p_{i}^{s'}\},\{p_{i}^{c'}\})=\mu(1-\Phi_{\mathsf{LB}})+(1-\mu)\Xi_{\mathsf{UB}}$, where $\mu\in[0,1]$ denotes the weighting factor to balance the two objective functions, and consider the following two branches:\footnote{If the magnitudes of two conflicting objectives are close to each other, normalization is not needed  \cite{Normalize}.} 
\begin{itemize}
     
\item If $\Psi(\mathbf{x}',\{p_{i}^{s'}\},\{p_{i}^{c'}\})\leq\Psi(\mathbf{x}^{[0]},\{p_{i}^{s[0]}\},\{p_{i}^{c[0]}\})$, we update $\mathbf{x}^{[1]}\leftarrow\mathbf{x}'$. Then we construct the neighborhood of $\mathcal{N}(\mathbf{x}^{[1]})$ to execute the next iteration; 
\item If $\Psi(\mathbf{x}',\{p_{i}^{s'}\},\{p_{i}^{c'}\})\!>\!\Psi(\mathbf{x}^{[0]},\{p_{i}^{s[0]}\},\{p_{i}^{c[0]}\})$, we execute the random flipping process repeatedly to generate another feasible solution within the $\mathcal{N}(\mathbf{x}^{[0]})$ until $\Psi(\mathbf{x}',\{p_{i}^{s'}\},\{p_{i}^{c'}\})\leq\Psi(\mathbf{x}^{[0]},\{p_{i}^{s[0]}\},\{p_{i}^{c[0]}\})$. 
\end{itemize}
This process is executed iteratively until the maximum number of iterations $\overline{I}$ is reached, generating a sequence $\{\mathbf{x}^{[1]},\mathbf{x}^{[2]},\cdots\}$. The entire algorithm is named fast-JPADS, and its complexity is $\mathcal{O}(\overline{I}\,N^{3.5})$.

\begin{figure*}[!t]
\centering
	\subfigure[]{
		\label{Doppler}
		\includegraphics[width=0.34\linewidth]{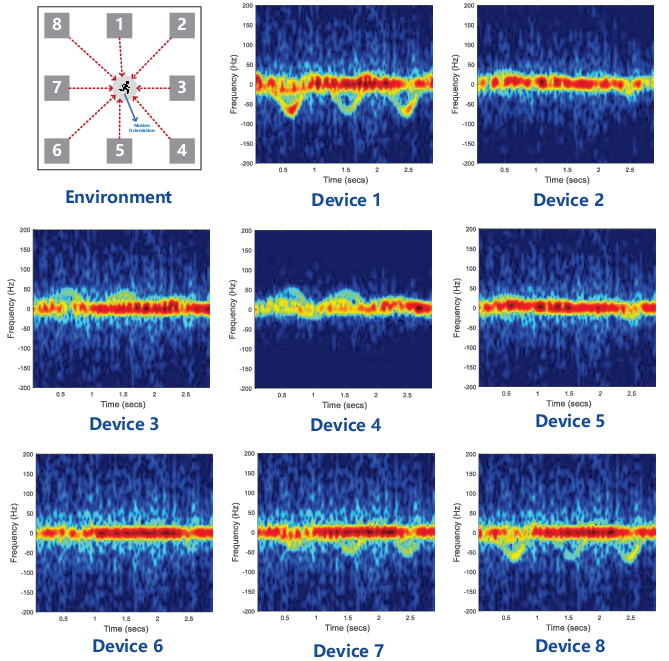}}
	\subfigure[]{
		\label{Iteration}
		\includegraphics[width=0.48\linewidth]{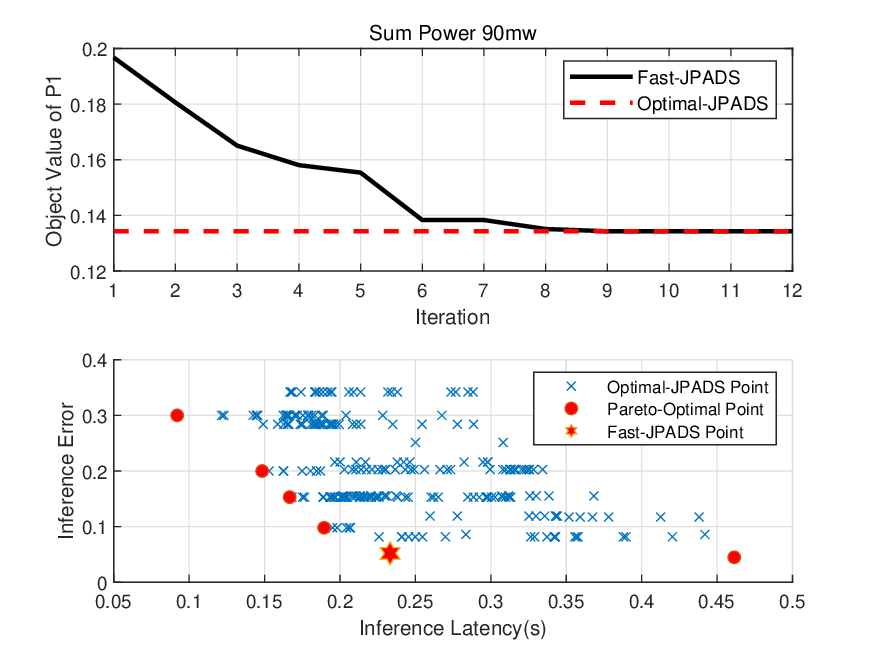}}
    \subfigure[]{
		\label{Accuracy}
		\includegraphics[width=0.45\linewidth]{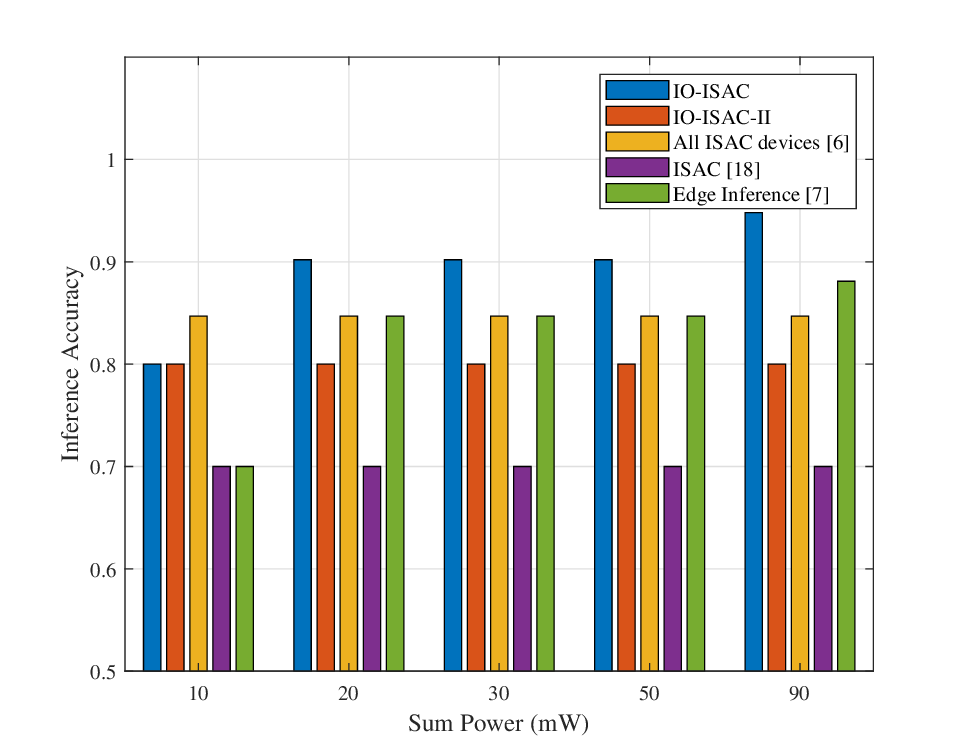}}
    \subfigure[]{
		\label{Latency}
		\includegraphics[width=0.45\linewidth]{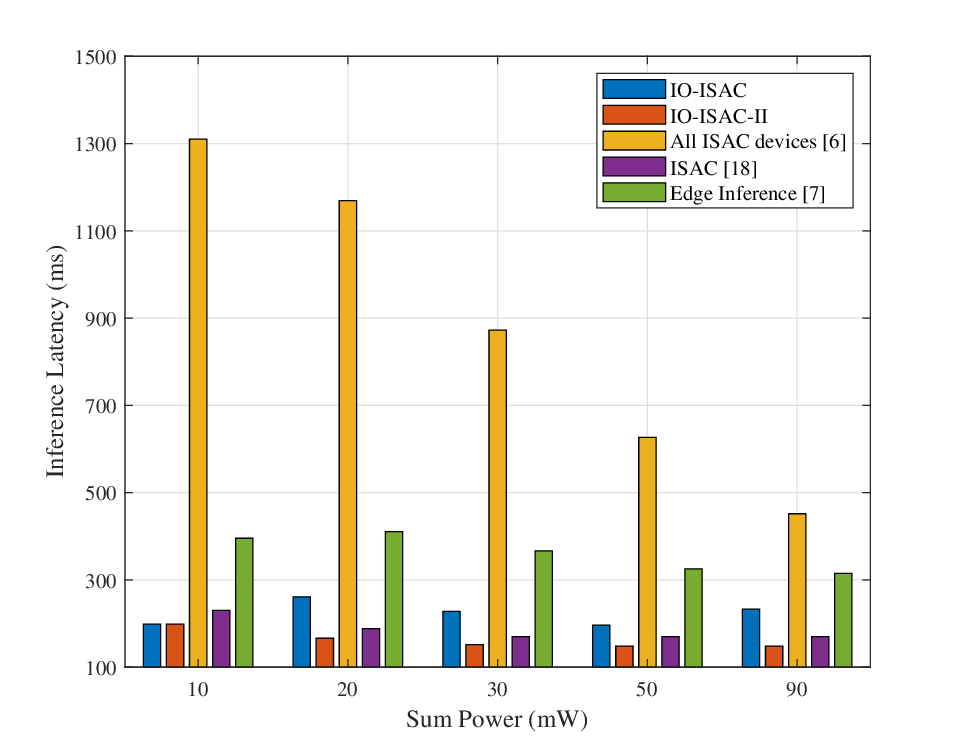}}
  \vspace{-0.15in}
	\caption{Simulation results. (a) Scenario and spectrograms. Devices 1 to 8 are located at $(0,5)$, $(10,5)$, $(10,0)$, $(-10,-5)$, $(0,-5)$, $(-10,-5)$, $(-10,0), (-10,5)$, respectively. Server and target are located at $(-2.5,10)$ and $(5,0)$, respectively. $P_{f}=P_{m}=0.1$ and $\lambda=3$ for all spectrograms; (b) The objective value of P1 versus the number of iterations for the fast-JPADS and the associated inference error-latency region; (c) Comparison of the inference accuracy between the IO-ISAC and benchmark schemes; (d) Comparison of the inference latency between the IO-ISAC and benchmark schemes.}
	\label{Entire_2}
\vspace{-0.2in}
\end{figure*}

\vspace{-10pt}
\section{Simulation Results}
This section presents simulation results to verify the performance of the IO-ISAC scheme. We simulate the case of $N=8$ in a conference hall with size $10\,\rm{m}\times20\,\rm{m}$, and the associated sensing spectrograms are shown in Fig.~\ref{Doppler}.
Unless otherwise specified, the maximum transmit power of each device is $P_T=30\,$mW and the total transmit power budget is $P_{\rm{sum}}=90\,$mW. 
The noise powers are set to $\sigma_{s}^2=-90\,$dBm and $\sigma_{c}^2=-60\,$dBm.
All channels are modeled as distance-dependent rician fading, with $30\,$dB pathloss at unit distance of $1\,$m. The data volume per sample is $0.125\,$MB.

First, the sum of two objectives in $\mathrm{P1}$ versus the number of iterations under the fast-JPADS algorithm is shown in Fig.~\ref{Iteration}. 
It can be seen that the fast-JPADS converges to a value close to the lower bound obtained by optimal-JPADS. 
The number of iterations required for the fast-JPADS to converge is around $10$. 
Consequently, we can early stop the iteration in practice, e.g., $\overline{I}=10$. 
On the other hand, the inference error-latency region of the fast-JPADS is shown in Fig.~\ref{Iteration}. It can be seen that the fast-JPADS achieves a region close to the Pareto front obtained by optimal-JPADS. 
In addition, the magnitudes of inference latency and error are close to each other.
This corroborates our analysis in Section III-C.
Due to its satisfactory performance, in the subsequent experiments, we adopt fast-JPADS for IO-ISAC scheme.

Next, to evaluate the performance gain brought by multi-view fusion, we compare the IO-ISAC scheme with that of ISAC \cite{Without_fusion}, which selects a single ISAC device (e.g., device $6$ in Fig.~\ref{Doppler}). It can be seen from Fig.~\ref{Accuracy} that the inference accuracy of IO-ISAC is $10\%$ to $24\%$ higher than that of the ISAC scheme. This is because the fluctuation of spectrogram becomes weak as the target's orientation is almost orthogonal to the sensing angle of device $6$ as shown in Fig.~\ref{Doppler}, resulting in subtle features for edge inference. In contrast, with the proposed IO-ISAC, the fluctuation of spectrograms is guaranteed to be strong, since there always exists a certain device outside the blind spot according to \textbf{Proposition 1}. For instance, at $P_{\rm{sum}}=10\,$mW, the IO-ISAC activates devices in $S=\{1,2\}$, and the two devices collaboratively eliminate the ``sensing blind spots''. We define the accuracy gap between IO-ISAC and ISAC as \textbf{multi-view fusion gain}. This gain increases with the sum power, since 
a higher view diversity can be exploited under a looser power constraint. For instance, at $P_{\rm{sum}}=30\,$mW and $90\,$mW, the number of activated devices becomes $4$ (i.e., $S=\{1,2,6,8\}$) and $6$ (i.e., $S=\{1,2,4,5,6,8\}$), respectively. We also provide a variation of IO-ISAC, denoted as IO-ISAC-2, which increases weight on the inference latency. Therefore, compared with IO-ISAC, IO-ISAC-2 has a smaller inference latency.


Then, to demonstrate the importance of device scheduling, we compare IO-ISAC with the all-devices scheme (i.e., $\mathbf{x}=[1,1,\ldots]^T$ with $|\mathcal{S}|\!=\!8$), which follows the design criteria in \cite{Aircomp}. It can be seen from Fig.~\ref{Accuracy} that 
the all-devices scheme leads to a lower inference accuracy than that of IO-ISAC although more devices are activated. This is because the co-device interference deteriorates the sensing SINRs at all devices. Moreover, as seen from Fig.~\ref{Latency}, the inference latency of the all-devices scheme exceeds $400\,$ms and is the highest among all the schemes. This is because different ISAC devices compete for the limited communication resources in the data uploading stage. In contrast, the proposed IO-ISAC selects a proper number of activated devices for collaborative inference, which results in a ``win-win'' situation that simultaneously reduces the error and latency. The gap between IO-ISAC and the all-devices scheme is \textbf{device scheduling gain}.

Finally, to evaluate the acceleration brought by ISAC, we compare the IO-ISAC scheme with the edge inference \cite{GNN_Edge}, which executes sensing and communication sequentially. It can be seen from Fig.~\ref{Accuracy} and Fig.~\ref{Latency} that the proposed IO-ISAC achieves an end-to-end latency significantly smaller than that of edge inference while guaranteeing similar inference accuracy. This demonstrates the benefits of one-stage reduction by merging the sensing and communication stages. We define this latency gap between IO-ISAC and edge inference as \textbf{ISAC acceleration gain}. Note that the ISAC co-functionality interference must be properly handled; otherwise, the gain may vanish.

\vspace{-5pt}
\section{Conclusion}
This paper has studied the end-to-end ISAC edge inference problem, and has derived two tractable models that bound the end-to-end edge inference error and the end-to-end edge inference latency from above. Leveraging these bounds, we have derived two JPADS algorithms and obtained the multi-view fusion gain, scheduling gain, and ISAC acceleration gain in an edge intelligence system. Simulation results have demonstrated the superiority of the IO-ISAC scheme.

\scriptsize
\bibliographystyle{IEEEtran}


\end{document}